\documentclass[11pt,a4paper,table]{article}
\usepackage[utf8]{inputenc}
\usepackage{enumerate}
\usepackage{amsmath}
\usepackage{dsfont}
\usepackage{amsfonts}
\usepackage{nicefrac}
\usepackage{amssymb}
\usepackage{amsthm}
\usepackage{mathtools}
\usepackage{tikz}
\usepackage{xcolor}
\usepackage{authblk}
\usepackage{tabularx}
\usepackage{paralist}
\usepackage{booktabs}
\usepackage{paralist}
\usepackage[pagebackref]{hyperref}
\usepackage[square,numbers]{natbib}
\usepackage[capitalize]{cleveref}
\usepackage[backgroundcolor=gray!10,textsize=footnotesize]{todonotes}
\usepackage{etoolbox}
\usepackage{xspace}
\usepackage{algpseudocode}
\usepackage{algorithm}
\usepackage{fullpage}
\usepackage{pifont}
\usepackage{graphicx}
\usepackage{calc}

\newtheorem{theorem}{Theorem}

\newtheorem{corollary}[theorem]{Corollary}

\theoremstyle{definition}

\newcommand{\cmark}{{\large \ding{51}}}
\newcommand{\xmark}{{\large \ding{55}}}

\newcommand{\N}{\mathbb{N}}

\newcommand{\Vor}{\mathrm{Vor}}
\newcommand{\rVor}{\mathrm{rVor}}
\newcommand{\Gg}{\mathcal{G}}
\newcommand{\td}{\mathrm{td}}

\Crefname{theorem}{Theorem}{Theorems}
\crefname{theorem}{Thm.}{Thms.}
\Crefname{corollary}{Corollary}{Corollaries}
\crefname{corollary}{Cor.}{Cors.}

\tikzset{
	vertex/.style={
		circle,
		scale=0.5,
		draw=black,
		fill = gray!7,
		thick,
                inner sep = 0pt,
                minimum size = 20pt
              },
              edge/.style={
                draw=black,
                very thick,
	}}

\title{Nash Equilibria in Reverse Temporal Voronoi Games}

\author[1]{Simeon Pawlowski}
\author[1]{Vincent Froese}

\affil[1]{\small
  Technische Universit\"at Berlin, Faculty~IV, Institute of Software Engineering and Theoretical Computer Science, Algorithmics and Computational Complexity.\protect\\
  \texttt{s.pawlowski@campus.tu-berlin.de},
  \texttt{vincent.froese@tu-berlin.de}}

\date{\today}

\begin{document}

\maketitle

\begin{abstract}
  We study Voronoi games on temporal graphs as introduced by Boehmer et al.~(IJCAI~'21) where two players each select a vertex in a temporal graph with the goal of reaching the other vertices earlier than the other player. In this work, we consider the \emph{reverse} temporal Voronoi game, that is, a player wants to maximize the number of vertices reaching her earlier than the other player.
  Since temporal distances in temporal graphs are not symmetric in general, this yields a different game.
  We investigate the difference between the two games with respect to the existence of Nash equilibria in various temporal graph classes including temporal trees, cycles, grids, cliques and split graphs. Our extensive results show that the two games indeed behave quite differently depending on the considered temporal graph class.
\end{abstract}

\section{Introduction}

The Voronoi game on graphs is an influence maximization game where two or more competitive players try to influence as many vertices as possible by choosing one initial vertex (or more) which then propagates the information to other vertices.
The Voronoi game on graphs was introduced by~\citet{DT07} and is motivated by modeling the spread of information (e.g.~viral marketing) or diseases within social networks.
Here, each player chooses an initial vertex and wins all vertices with a shorter distance to her chosen vertex than to any other player.
A central game-theoretic question is the existence of a Nash equilibrium, that is, a stable strategy profile where no player has an incentive to deviate, for a Voronoi game on a given graph. This question has been studied for Voronoi games on various classes of graphs such as trees, cycles and grids~\cite{MMPS08,FMM09,SSXZ20}.

Recently, \citet{BFHLNR21} introduced the \emph{temporal} Voronoi game which is played on temporal graphs, that is, graphs where the edge set changes over discrete time steps~\cite{Michail16}. Due to their dynamic nature, temporal graphs are a more realistic model for social networks and are thus a natural extension for the Voronoi game. \citet{BFHLNR21} defined the payoff of a player to be the number of vertices which she reaches earlier (that is, the temporal distance is smaller) than any other player and studied Nash equilibria on various forms of temporal paths, trees and cycles. They posed the question how the temporal Voronoi game behaves for other temporal distance notions.
Note that the temporal distance between vertices in a temporal graph is not symmetric. Hence,
we study the reverse definition where a player gains the vertices that reach her before any other player.

\paragraph*{Related Work.}
The Voronoi game has originally been introduced as a competitive facility location problem in continuous spaces by~\citet{ACCGO04}.
For the discrete Voronoi game on graphs, also complexity-theoretic questions about existence of Nash equilibria have been studied \cite{DT07,TDU11,BBDS15}.
Another similar game on graphs is the \emph{competitive diffusion game} which was introduced by \citet{AFPT10} and for which Nash equilibria
have also been studied on various graph classes \cite{Rosh14,BFT16,FHKO22}.
\citet{BFHLNR21} also studied diffusion games on temporal graphs.

\paragraph*{Our Contributions.}
We study the existence of Nash equilibria in \emph{reverse} temporal Voronoi games with two players on (among others) temporal trees, cycles, cliques, grids, and split graphs (see \Cref{tab:results} for an overview). Our results answer the question of guaranteed existence of Nash equilibria for a wide range of temporal graph classes.
For the sake of completeness, we also obtain results for the ``classic'' temporal Voronoi game on the corresponding temporal graph classes.
It turns out that the two games indeed behave differently on temporal trees, cycles and split graphs.
For example, on temporally connected trees there is always a Nash equilibrium in the reverse game but not in the classic game, while on monotonically growing cycles the opposite is true.
One of the key differences between the two games is that in the classic game the two players can ``catch up'' each other while this effect does not exist for the reverse game.
This seemingly renders the reverse temporal Voronoi game easier to analyze than the temporal Voronoi game where the catch-up dynamics can cause more complicated situations.
An interesting side-observation is that on temporally connected trees the reverse temporal Voronoi game behaves similar to the static Voronoi game on static trees while on monotonically shrinking split graphs the classic temporal Voronoi game behaves analogous to the static case.

\newcommand{\yes}{\makebox[2em][c]{\cmark}}
\newcommand{\yesPar}{\makebox[2em][c]{(\hspace{-.5pt}\cmark\hspace{-.5pt})}}
\newcommand{\yesRef}[1]{\yes (\Cref{#1})}
\newcommand{\no}{\makebox[2em][c]{\xmark}}
\newcommand{\noPar}{\makebox[2em][c]{(\hspace{.67pt}\xmark\hspace{.67pt})}}
\newcommand{\noRef}[1]{\no (\Cref{#1})}
\newcommand{\acc}{\rowcolor{gray!10}}

\begin{table}[t]
	\centering
	\footnotesize
        \def\arraystretch{1.3}
        \caption{Overview of results. The results in the top rows (marked with $^*$) are by \citet{BFHLNR21}.
          A ``\cmark''~indicates guaranteed existence of a Nash equilibrium while
          an ``\xmark''~means that a Nash equilibrium is not guaranteed to exist. 
          Entries in parentheses are implied by other table cells.\vspace{7pt}
        }
	\belowrulesep = 5pt
	\setlength{\tabcolsep}{2.5mm}
	\begin{tabular}{l@{\hspace{5.5mm}}l@{\hspace{5.5mm}}l@{\hspace{5.5mm}}l}
          \toprule
          							& \textbf{Temporally}	& \textbf{Monotonically}		& \textbf{Monotonically}	       	\\
          							& \textbf{connected}	& \textbf{growing}\vspace{3pt}	& \textbf{shrinking}	       	\\
          \midrule[\heavyrulewidth]
          \textbf{Temporal Voronoi}\vspace{5pt}														       	\\
          \acc Temporal Paths$^*$			& \no				& \yesPar				& \no		       			\\
          Temporal Trees$^*$			& \noPar       			& \yes					& \noPar			       	\\
          \acc Temporal Cycles$^*$ 		& \no				& \yes					& \no			       		\\
          Temporal Grids				& \noPar       			& \noRef{thm:Vor-grid}		& \noPar			       	\\
          \acc Temporal Cliques      			& \noPar       			& \noRef{cor:Vor-clique}		& \yesPar	       			\\
          Temp. Complete $k$-partite ($k\ge 2$)	& \noPar       			& \noRef{cor:Vor-clique}		& \yesRef{thm:Vor-k-partit} 	\\
          \acc Temporal Threshold			& \noPar       			& \noPar       				& \yesPar			       	\\
          Temporal Split				& \noPar       			& \noPar       				& \yesRef{thm:Vor-split}		\\
          \midrule[\heavyrulewidth]
          \textbf{Reverse Temporal  Voronoi}\vspace{5pt}										   	           	\\
          \acc Temporal Paths				& \yesPar			& \yesPar				& \noRef{thm:path}		\\
          Temporal Trees				& \yesRef{thm:tree}	& \yesPar				& \noPar       				\\
          \acc Temporal Cycles         			& \noPar       			& \noRef{thm:cycle-grow}  	& \noRef{thm:cycle-shrink}	\\
          Temporal Grids 				& \noPar       			& \noRef{thm:grid}			& \noPar       				\\
          \acc Temporal Cliques      			& \noPar       			& \noRef{thm:clique-grow} 	& \yesPar				\\
          Temp. Complete $k$-partite ($k\ge 2$)	& \noPar       			& \noRef{thm:k-part-grow}	& \yesRef{thm:k-part-shrink}	\\
          \acc Temporal Threshold			& \noPar       			& \noPar       				& \yesRef{thm:k-part-shrink}	\\
          Temporal Split				& \noPar       			& \noPar       				& \noRef{thm:split} 		\\
          \bottomrule
	\end{tabular}
	\label{tab:results}
\end{table}

\paragraph*{Organization of the Paper.}
Our paper is organized as follows: \Cref{sec:prel} introduces basic definitions of temporal graphs and the (reverse) temporal Voronoi game.
The results for the reverse temporal Voronoi game are then presented in \Cref{sec:rVor} followed by the results for the temporal Voronoi game in \Cref{sec:Vor}.
We conclude with some open questions in \Cref{sec:conc}.

\section{Preliminaries} \label{sec:prel}
For $a \le b\in \mathbb{N}$, let $[a,b]:=\{a,a+1,\dots, b\}$ and let $[a]:=[1,a]$.

\paragraph*{Temporal Graphs.}
A \emph{temporal graph}~$\mathcal{G} = (V, (E_t)_{t=1}^\infty)$
consists of a finite set~$V$ of vertices
and an infinite sequence~$(E_t)_{t=1}^\infty$ of edge sets 
$E_t\subseteq \binom{V}{2}$.
If there is an integer~$i$ such that $E_t = E_i$ for all $t \geq i$,
then we define the \emph{lifetime}~$\tau(\Gg)$ of~$\Gg$ to be the minimum such integer.
For our considered game, we can assume that all temporal graphs have finite lifetime~$\tau$.
Hence, we do not specify $E_i$ for $i > \tau$.
The (static) graph~$G_t\coloneqq(V,E_t)$ is called the \emph{$t$-th layer} of~$\mathcal{G}$
and $\mathcal{G}_\downarrow \coloneqq 
(V, E_\downarrow)$ with~$E_\downarrow := \bigcup_{t=1}^\infty E_t$
is the \textit{underlying (static) graph} of~$\mathcal{G}$.

A \emph{temporal path} (resp. \emph{tree}, \emph{cycle} etc.) is a temporal graph
whose underlying graph is a path (resp.\ tree, cycle etc.).
A (static) \emph{$(n\times m)$-grid graph} is a graph which is isomorphic to~$([n]\times[m],\{\{(i,j),(i',j')\}\mid |i-i'|+|j-j'|=1\})$.
A \emph{split graph} is a graph whose vertex set can be partitioned into a clique and an independent set.
A \emph{threshold graph} is a split graph that can be constructed by iteratively adding isolated vertices or dominating vertices
(hence, there always exists a vertex which dominates all non-isolated vertices).

In a  temporal graph $\mathcal{G}=(V,(E_t)_{t=1}^\infty)$, a \emph{temporal walk} from a vertex $u$ to a vertex $v$ is a sequence $(\{v_0\coloneqq u,v_1\},t_1),(\{v_1,v_2\},t_2),\dots, (\{v_{d-1},v_d\coloneqq v\},t_d)$ such that $t_i < t_{i+1}$ for all $i\in[d-1]$ and $\{v_{i-1},v_i\}\in E_{t_i}$ for all $i\in[d]$.
We call $t_d$ the \emph{arrival time} of the temporal walk.
A temporal walk from~$u$ to~$v$ is called \emph{foremost} if there 
is no temporal walk from~$u$ to~$v$ with earlier arrival time.
The \emph{temporal distance} $\td(u,v)$ from $u$ to $v$ is the arrival 
time of a foremost walk from $u$ to~$v$ (we set~$\td(u, v) \coloneqq 0$ if $u=v$).
If there is no such walk, then $\td(u, v) \coloneqq \infty$.
Note that temporal distances are not symmetric, that is, $\td(u,v)\neq \td(v,u)$ is possible.
We say that a vertex \emph{$u$ reaches} a vertex $v$ \emph{until} (\emph{at}) step $t$ if $\td(u,v)\leq t$ ($=t$).

A temporal graph~$\Gg$ is \emph{temporally connected} if $\td(u, v) < \infty$ for all vertex pairs $u, v$.
Further, $\Gg$ is \emph{monotonically growing (shrinking)} 
if edges do not disappear (appear) over time, that is, $E_t\subseteq E_{t+1}$ ($E_{t+1}\subseteq E_t$) for all~$t$.
Note that, if $\Gg_\downarrow$ is connected, then monotonic growth of~$\Gg$ implies temporal connectedness.

\paragraph{(Reverse) Temporal Voronoi Games.}
For a temporal graph $\mathcal{G}=(V,(E_t)_{t=1}^\infty)$ and a number~$k\in \mathbb{N}$ 
of players, \citet{BFHLNR21} introduced the 
$k$-player \emph{temporal Voronoi game} $\Vor(\Gg,k)$ on $\mathcal{G}$.
The \emph{strategy space} of each player~$i\in [k]$ is the vertex set~$V$, that is,
each player~$i$ selects a single vertex $p_i\in V$ (also called \emph{position}).
A \textit{strategy profile} is a tuple~$(p_1 , \dots, p_k ) \in 
V^k$ containing the chosen vertices of all players.
In $\Vor(\Gg,k)$, the strategy profile $(p_1 , \dots, p_k )$ determines a vertex subset $U_i(p_1,\ldots,p_k)\coloneqq\{v\in V\mid  \forall j\neq i:\td(p_i,v)<\td(p_j,v)\}$ for each player~$i$.
The \emph{payoff} of player~$i$ is then $u_i(p_1,\ldots,p_k)\coloneqq |U_i(p_1, \dots, p_k)|$.
That is, each player ``wins'' those vertices which she reaches earlier than all other players.
In the \emph{reverse temporal Voronoi game} $\rVor(\Gg,k)$, we define the set $U_i(p_1,\ldots,p_k)\coloneqq\{v\in V\mid  \forall j\neq i:\td(v,p_i)<\td(v,p_j)\}$, that is,
each player ``wins'' those vertices which reach her earlier than any other player.

In both games, the players aim to maximize their payoffs.
Hence, player~$i$ plays a \emph{best response} to the other 
players in~$(p_1, \dots, p_k)$ if for all
vertices~$p' \in V$ it holds that \[u_i(p_1,\dots, p_{i-1}, p', p_{i+1}, \dots, p_k) \leq u_i(p_1, \dots, p_k).\]
A strategy profile~$(p_1, \dots, p_k)$ is a \textit{Nash equilibrium} if every player plays a best response to the other players.
In this paper, we only consider $k=2$ players.

\section{Reverse Temporal Voronoi Games ($\rVor$)}\label{sec:rVor}

In this section we prove the results for the reverse temporal Voronoi game shown in \Cref{tab:results}.

\subsection{Temporally Connected Graphs}

For temporally connected graphs, a Nash equilibrium is only guaranteed if the underlying graph is a tree.
In fact the Nash equilibrium is analogous to the Voronoi game on static graphs.
Note that for all other temporal graph classes considered in this paper, there are examples without a Nash equilibrium already for monotonically growing graphs (as shown in \Cref{sec:mon_grow}).

\begin{theorem}\label{thm:tree}
  On every temporally connected tree $\mathcal T$, there exists a Nash equilibrium in $\rVor(\mathcal{T},2)$.
\end{theorem}

\begin{proof}
  Let~$T\coloneqq\mathcal{T}_\downarrow=(V,E)$ with~$|V|\ge 2$ (the case~$|V|=1$ is trivial). Let~$p_1$ be a centroid of~$T$ (that is, a vertex~$v$ that minimizes the maximum size of any connected component in~$T-v$) and let~$p_2$ be a neighbor of~$p_1$ in a maximum-size component~$C=(V',E')$ of~$T-p_1$. Then $(p_1,p_2)$ is a Nash equilibrium. Note that~$U_1(p_1,p_2)=V\setminus V'$ since $\mathcal{T}$ is temporally connected and all vertices in~$V\setminus V'$ reach~$p_1$ before~$p_2$. Analogously, it holds $U_2(p_1,p_2)=V'$. Since~$p_1$ is a centroid, we have $u_1(p_1,p_2)\ge |V|/2$ and~$u_2(p_1,p_2)\le |V|/2$. Clearly, player~2 cannot improve since she could only win vertices within a component of~$T-p_1$ and~$C$ is already maximal. Also player~1 cannot improve since she could only win a subset of vertices of~$V\setminus V'$ or~$V'$.
\end{proof}

\subsection{Monotonically Growing Graphs}\label{sec:mon_grow}

In the following, we show that disallowing edges to disappear does not guarantee a Nash equilibrium (except for trees). The following theorem is in contrast to the classic temporal Voronoi game where a Nash equilibrium always exists for monotonically growing cycles as shown by \citet{BFHLNR21}.

\begin{theorem}\label{thm:cycle-grow}
  There is a monotonically growing cycle $\mathcal C$ such that there is no Nash equilibrium in $\rVor(\mathcal{C},2)$.
\end{theorem}

\begin{proof}
  Consider the temporal cycle $\mathcal{C}\coloneqq ([7],E_1,E_2)$ where $E_1\coloneqq\{\{i,i+1\}\mid i\in[6]\}\setminus \{\{2,3\}\}$ and $E_2\coloneqq E_1\cup \{\{2,3\},\{7,1\}\}$ (see \Cref{fig:cycle}).
  To show that there is no Nash equilibrium, we show that both players can always win at least~4 vertices regardless of the choice of the other player. Since there are only~7 vertices in total, there cannot be a Nash equilibrium.
  To show the above claim, we assume that~$p_1\in[2,5]$ (by symmetry of~$\mathcal{C}$).
  The following cases are easily verified:
  \begin{compactitem}
    \item If~$p_1=2$, then, for~$p_2=5$, it holds $[4,7]\subseteq U_2(p_1,p_2)$.
    \item If~$p_1=3$, then, for~$p_2=4$, it holds $[4,7]\subseteq U_2(p_1,p_2)$.
    \item If~$p_1=4$, then, for~$p_2=7$, it holds $\{1,2,6,7\}\subseteq U_2(p_1,p_2)$.
    \item If~$p_1=5$, then, for~$p_2=4$, it holds $[1,4]\subseteq U_2(p_1,p_2)$.
  \end{compactitem}
\end{proof}

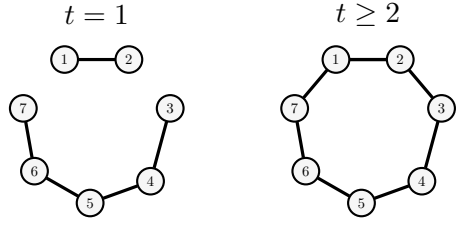
\begin{figure}[t]
  \center 
  \begin{tikzpicture}[scale=1]
    \node[vertex] (1) at (115:1) {1};
    \node[vertex] (2) at (65:1) {2};
    \node[vertex] (3) at (15:1) {3};
    \node[vertex] (4) at (315:1) {4};
    \node[vertex] (5) at (265:1) {5};
    \node[vertex] (6) at (215:1) {6};
    \node[vertex] (7) at (165:1) {7};
    \draw[edge] (1) -- (2);
    \draw[edge] (3) -- (4) -- (5) -- (6) -- (7);
    \node at (90:1.5) {$t=1$};
  \end{tikzpicture}
  \hspace{1cm}
  \begin{tikzpicture}[scale=1]
    \node[vertex] (1) at (115:1) {1};
    \node[vertex] (2) at (65:1) {2};
    \node[vertex] (3) at (15:1) {3};
    \node[vertex] (4) at (315:1) {4};
    \node[vertex] (5) at (265:1) {5};
    \node[vertex] (6) at (215:1) {6};
    \node[vertex] (7) at (165:1) {7};
    \draw[edge] (1) -- (2);
    \draw[edge] (1) -- (2) -- (3) -- (4) -- (5) -- (6) -- (7) -- (1);
    \node at (90:1.5) {$t\ge 2$};
  \end{tikzpicture}
  \caption{A monotonically growing temporal cycle without a Nash equilibrium.}
  \label{fig:cycle}
\end{figure}
\noindent
Recall that a Nash equilibrium is guaranteed for the temporal Voronoi game on monotonically growing cycles.
An example for a Nash equilibrium in the temporal Voronoi game~$\Vor(\mathcal{C},2)$ on the cycle from the proof above, is $(5,4)$. Here, both players win three vertices. 

From \Cref{thm:cycle-grow}, we easily obtain an analogous result for monotonically growing cliques (and with it also for split and threshold graphs.).

\begin{corollary}\label{thm:clique-grow}
  There is a monotonically growing clique~$\mathcal{Q}$ such that there is no Nash equilibrium in $\rVor(\mathcal{Q},2)$.
\end{corollary}

\begin{proof}
  Consider the monotonically growing cycle~$\mathcal{C}$ from \Cref{thm:cycle-grow} where no Nash equilibrium exists. Since~$\mathcal{C}$ is temporally connected, all pairwise temporal distances are finite, that is, $\td(u,v)\le d$ for all~$u,v\in [7]$ and some~$d\in\N$.
  Hence, any further edge appearing after time~$d$ is not changing the temporal distance of any vertex pair.
  Therefore,~$\mathcal{Q}\coloneqq ([7],E'_1,\ldots,E'_{d+1})$ with~$E'_t\coloneqq E_t$ for all~$t\le d$ and $E'_{d+1}\coloneqq \binom{[7]}{2}$ is a monotonically growing clique without a Nash equilibrium.
\end{proof}

Next, we show that also on monotonically growing grids there is no guarantee for a Nash equilibrium.

\begin{theorem}\label{thm:grid}
  There is a monotonically growing grid $\mathcal G$ such that there is no Nash equilibrium in $\rVor(\mathcal{G},2)$.
\end{theorem}

\begin{proof}
  Consider the graph~$\Gg\coloneqq ([6],E_1,E_2)$ (depicted in~\Cref{fig:grid}) with
  \begin{align*}
    E_1&\coloneqq\{\{1,2\},\{1,4\},\{3,6\},\{5,6\}\}\text{ and}\\
    E_2&\coloneqq E_1\cup \{\{2,3\},\{2,5\},\{4,5\}\}.         
  \end{align*} 
  By symmetry, assume that~$p_1\in\{1,2,4\}$.
  \begin{compactitem}
    \item If $p_1=1$, then the best response by player~2 is~$p_2=2$, where~$U_1(1,2)=\{1,4\}$ and~$U_2(1,2)=\{2,3,5,6\}$. But~$p_1=6$ yields~$U_1(6,2)=\{3,5,6\}$.
    \item If $p_1=2$, then~$p_2=6$ is the best response with~$U_1(2,6)=\{1,2,4\}$ and~$U_2(2,6)=\{3,5,6\}$.
  But then~$p_1=5$ is the best response with~$U_1(5,6)=\{1,2,4,5\}$.
    \item If~$p_1=4$, then a best response is~$p_2=6$ with~$U_1(4,6)=\{1,2,4\}$ and~$U_2(4,6)=\{3,5,6\}$.
      But then the best response is~$p_1=5$ again.
      Also $p_2=5$ is a best response, where~$U_1(4,5)=\{1,4\}$ and~$U_2(4,5)=\{3,5,6\}$.
      Then, $p_1$ can improve with $U_1(1,5)=\{1,2,4\}$.
      \end{compactitem}
\end{proof}
\noindent
Interestingly, note that (1,6) is a Nash equilibrium for the classic temporal Voronoi game on the above grid of \Cref{thm:grid}. (However, this is not the case for all monotonically growing grids as shown in \Cref{thm:Vor-grid}.)

Again, from \Cref{thm:grid}, we easily obtain the following corollary.

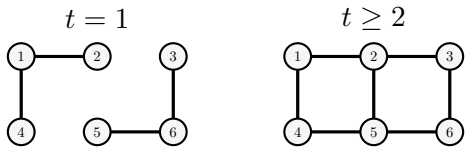
\begin{figure}[t]
  \center 
  \begin{tikzpicture}[scale=1]
    \node[vertex] (1) at (0,1) {1};
    \node[vertex] (2) at (1,1) {2};
    \node[vertex] (3) at (2,1) {3};
    \node[vertex] (4) at (0,0) {4};
    \node[vertex] (5) at (1,0) {5};
    \node[vertex] (6) at (2,0) {6};
    \draw[edge] (4) -- (1) -- (2);
    \draw[edge] (5) -- (6) -- (3);
    \node at (1,1.5) {$t=1$};
  \end{tikzpicture}
  \hspace{1cm}
  \begin{tikzpicture}[scale=1]
    \node[vertex] (1) at (0,1) {1};
    \node[vertex] (2) at (1,1) {2};
    \node[vertex] (3) at (2,1) {3};
    \node[vertex] (4) at (0,0) {4};
    \node[vertex] (5) at (1,0) {5};
    \node[vertex] (6) at (2,0) {6};
    \draw[edge] (1) -- (2) -- (3) -- (6) -- (5) -- (4) -- (1);
    \draw[edge] (2) -- (5);
    \node at (1,1.5) {$t\ge 2$};
  \end{tikzpicture}
  \caption{A monotonically growing temporal grid without a Nash equilibrium.}
  \label{fig:grid}
\end{figure}

\begin{corollary}\label{thm:k-part-grow}
  For every~$k\ge 2$, there is a monotonically growing complete $k$-partite graph $\mathcal K$ such that there is no Nash equilibrium in $\rVor(\mathcal{K},2)$.
\end{corollary}

\begin{proof}
  Consider the monotonically growing grid~$\Gg$ from \Cref{thm:grid} (which is~bipartite) where no Nash equilibrium exists. Since all pairwise temporal vertex distances are finite (at most some~$d\in\N$), we can modify~$\Gg$ as follows without introducing a Nash equilibrium:
  $\mathcal{K}\coloneqq ([4+k],E'_1,\ldots,E'_{d+1})$, where~$E'_t\coloneqq E_t$ for all~$t\le d$ and \[E'_{d+1}\coloneqq E'_d\cup\{\{1,6\},\{3,4\}\}\cup\bigcup_{j=7}^{4+k}\{\{j,i\}\mid i < j\}.\]
  Note that~$\mathcal{K}$ is a monotonically growing complete $k$-partite graph where all temporal distances between vertices in~$[6]$ are the same as in~$\Gg$. It remains to check that there is no Nash equilibrium. If both players pick vertices in~$[6]$, then the outcome is exactly the same as in~$\Gg$ (all newly introduced vertices have equal temporal distance to both players and hence are not won by any player). Hence, this is not a Nash equilibrium. If a player picks one of the new vertices, then this is never optimal, since she only wins this single vertex, whereas she could win at least two vertices by choosing some vertex in~$[6]$.
\end{proof}

\subsection{Monotonically Shrinking Graphs}

We now consider temporal graphs where no edges are allowed to appear over time.
It turns out that among the graph classes we considered, a Nash equilibrium is only guaranteed if the game is essentially ``decided'' in the first layer, that is, on temporal complete $k$-partite and threshold graphs.
We start with excluding Nash equilibria from all other considered temporal graph classes.

\begin{theorem}\label{thm:path}
  There is a monotonically shrinking path~$\mathcal P$ such that there is no Nash equilibrium in $\rVor(\mathcal P,2)$.
\end{theorem}

\begin{proof}
  Let~$\mathcal P \coloneqq ([9], E_1,E_2)$ with~$E_1\coloneqq \{\{i,i+1\}\mid i\in[8]\}$ and~$E_2\coloneqq E_1\setminus\{\{3,4\}\}$ (\Cref{fig:path}). Clearly, if~$\{p_1,p_2\}\subseteq [3]$, then this is not optimal, since each player can win at most three vertices, whereas choosing vertex~4 yields at least six vertices.
  If~$\{p_1,p_2\}\subseteq [4,9]$, then we can assume without loss of generality that~$4\le p_1< p_2=p_1+1$.
  \begin{compactitem}
    \item If~$p_1= 4$, then~$u_1(p_1,p_2)=2$ and~$u_1(6,p_2)=4$.
    \item If~$p_1= 5$, then~$u_1(p_1,p_2)=3$ and~$u_1(3,p_2)=4$.
    \item If~$p_1\ge 6$, then~$u_2(p_1,p_2)\le 3$ and~$u_2(p_1,3)\ge 4$.
  \end{compactitem}
  Finally, let~$p_1\le 3 < p_2$ (wlog).
  \begin{compactitem}
    \item If~$p_2\le 5$, then~$u_1(p_1,p_2)\le 3$ and~$u_1(6,p_2)=4$.
    \item If~$p_2> 5$ and~$4 \in U_1(p_1,p_2)$, then~$u_2(p_1,p_2)\le 5$ and~$u_2(p_1,4)=6$.
    \item If~$p_2>5$ and~$4 \not\in U_1(p_1,p_2)$, then~$u_1(p_1,p_2)\le 3$ and~$u_1(3,p_2)=4$.
  \end{compactitem}
\end{proof}

\begin{figure}[t]
  \center 
  \begin{tikzpicture}[scale=1]
    \node[vertex] (1) at (0,0) {1};
    \node[vertex] (2) at (1,0) {2};
    \node[vertex] (3) at (2,0) {3};
    \node[vertex] (4) at (3,0) {4};
    \node[vertex] (5) at (4,0) {5};
    \node[vertex] (6) at (5,0) {6};
    \node[vertex] (7) at (6,0) {7};
    \node[vertex] (8) at (7,0) {8};
    \node[vertex] (9) at (8,0) {9};
    \draw[edge] (1) -- (2) -- (3) -- (4) -- (5) -- (6) -- (7) -- (8) -- (9);
    \node at (-1,0) {$t=1$};
    \node[vertex] (1) at (0,-1) {1};
    \node[vertex] (2) at (1,-1) {2};
    \node[vertex] (3) at (2,-1) {3};
    \node[vertex] (4) at (3,-1) {4};
    \node[vertex] (5) at (4,-1) {5};
    \node[vertex] (6) at (5,-1) {6};
    \node[vertex] (7) at (6,-1) {7};
    \node[vertex] (8) at (7,-1) {8};
    \node[vertex] (9) at (8,-1) {9};
    \draw[edge] (1) -- (2) -- (3);
    \draw[edge] (4) -- (5) -- (6) -- (7) -- (8) -- (9);
    \node at (-1,-1) {$t\ge 2$};
  \end{tikzpicture}
  \caption{A monotonically shrinking temporal path without a Nash equilibrium.}
  \label{fig:path}
\end{figure}
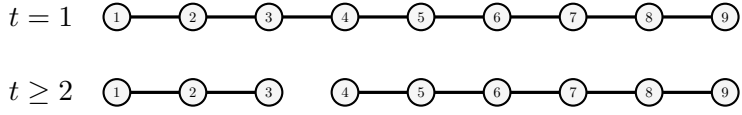

\begin{theorem}\label{thm:cycle-shrink}
  There is a monotonically shrinking cycle $\mathcal C$ such that there is no Nash equilibrium in $\rVor(\mathcal{C},2)$.
\end{theorem}

\begin{proof}
  Let $\mathcal C\coloneqq ([10],E_1,E_2)$ with $E_1\coloneqq\{\{i,i+1\}\mid i\in[9]\}\cup\{\{10,1\}\}$ and $E_2\coloneqq E_1\setminus\{\{3,4\},\{10,1\}\}$ (see \Cref{fig:cycle_shrink}).
  Clearly, if~$\{p_1,p_2\}\subseteq [3]$, then this is not a Nash equilibrium since each player wins at most four vertices and choosing vertex~4 would yield at least six vertices.
  Now consider the case if~$\{p_1,p_2\}\subseteq [4,10]$. Note that vertex~2 does not reach any of the players. Hence, the remaining graph behaves like a path. Therefore, we can assume that~$p_2=p_1+1$.
  \begin{compactitem}
    \item If~$p_1\ge 8$, then~$u_2(p_1,p_2)\le 3$ and~$u_2(p_1,7)\ge 5$.
    \item If~$p_1 \le 5$, then~$u_1(p_1,p_2)\le 3$ and~$u_1(7,p_2)\ge 5$.
    \item If~$p_1 =6$, then~$u_1(6,7)=4$ and~$u_1(2,7)=5$.
    \item If~$p_1 =7$, then~$u_2(7,8)=4$ and~$u_2(7,2)=5$.
  \end{compactitem}
  Finally, assume~$p_1 \le 3 < p_2 \le 7$ (by symmetry of the cycle).
  \begin{compactitem}
    \item If~$p_2\le 5$, then~$u_1(p_1,p_2)\le 4$ and~$u_1(6,p_2)=6$.
    \item If~$p_2 = 6$, then~$U_2(p_1,6)\subseteq[4,9]$. If~$4\not\in U_2(p_1,6)$, then~$p_1\ge 2$ and thus player~2 can improve with~$p_2=4$ giving~$U_2(p_1,4)=[4,9]$.
      If~$4\in U_2(p_1,6)$, then~$p_1=1$ and player~2 can improve with~$p_2=4$ to~$U_2(1,4)=[3,9]$.
    \item If~$p_2 =7$, then~$u_2(p_1,7)\le 5$ and~$u_2(p_1,4)\ge 6$.
  \end{compactitem}
\end{proof}

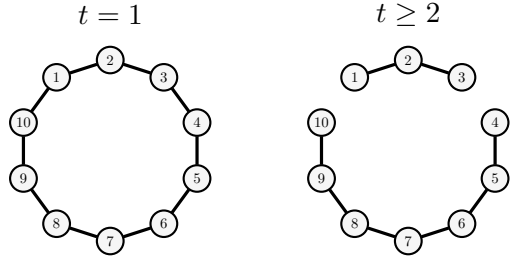
\begin{figure}[t]
  \center 
  \begin{tikzpicture}[scale=1.2]
    \node[vertex] (1) at (126:1) {1};
    \node[vertex] (2) at (90:1) {2};
    \node[vertex] (3) at (54:1) {3};
    \node[vertex] (4) at (18:1) {4};
    \node[vertex] (5) at (342:1) {5};
    \node[vertex] (6) at (306:1) {6};
    \node[vertex] (7) at (270:1) {7};
    \node[vertex] (8) at (234:1) {8};
    \node[vertex] (9) at (198:1) {9};
    \node[vertex] (10) at (162:1) {10};
    \draw[edge] (1)--(2)--(3) -- (4) -- (5) -- (6) -- (7)--(8)--(9)--(10)--(1);
    \node at (90:1.5) {$t=1$};
  \end{tikzpicture}
  \hspace{1cm}
  \begin{tikzpicture}[scale=1.2]
    \node[vertex] (1) at (126:1) {1};
    \node[vertex] (2) at (90:1) {2};
    \node[vertex] (3) at (54:1) {3};
    \node[vertex] (4) at (18:1) {4};
    \node[vertex] (5) at (342:1) {5};
    \node[vertex] (6) at (306:1) {6};
    \node[vertex] (7) at (270:1) {7};
    \node[vertex] (8) at (234:1) {8};
    \node[vertex] (9) at (198:1) {9};
    \node[vertex] (10) at (162:1) {10};
    \draw[edge] (1)--(2)--(3);
    \draw[edge] (4) -- (5) -- (6) -- (7)--(8)--(9)--(10);
    \node at (90:1.5) {$t\ge 2$};
  \end{tikzpicture}
  \caption{A monotonically shrinking temporal cycle without a Nash equilibrium.}
  \label{fig:cycle_shrink}
\end{figure}
\noindent
Notably, for temporal paths already one disappearing edge is enough to exclude a Nash equilibrium while the counterexample for cycles has two disappearing edges.
In fact, one can show that for cycles a Nash equilibrium always exists if at most one edge disappears.

It remains to exclude Nash equilibria for monotonically shrinking split graphs.
\begin{theorem}\label{thm:split}
  There exists a monotonically shrinking split graph $\mathcal S$ such that there is no Nash equilibrium in $\rVor(\mathcal{S},2)$.
\end{theorem}

\begin{proof}
  Let~$\mathcal S\coloneqq (V=C\cup I, E_1,E_2)$ with~$C\coloneqq[4,7]$, $I\coloneqq\{1,2,3,8\}$, and
  \begin{align*}
    E_1&\coloneqq\binom{C}{2}\cup\{\{1,4\},\{2,4\},\{2,5\},\{3,5\},\{6,8\},\{7,8\}\},\\
    E_2&\coloneqq\{\{2,4\},\{2,5\},\{4,6\},\{5,7\}\}.
  \end{align*}
  \Cref{fig:split} shows~$\mathcal S$. By symmetry of~$\mathcal S$, let~$p_1\in\{1,2,4,6,8\}$.
  \begin{compactitem}
    \item If~$p_1=1$, then clearly~$p_2=4$ is the best response. But player~1's best response is then~$p_1=7$ which yields~$U_1(7,4)=\{3,7,8\}$.
    \item If~$p_1=2$, then we can assume~$p_2\in\{1,4,6,8\}$ by symmetry. The best response is~$p_2=4$, where~$U_1(2,4)=\{2,3\}$. Again, player~1 can improve with~$p_1=7$.
    \item If~$p_1=4$, then~$p_2=7$ is the best response with~$U_1(4,7)=\{1,2,4\}$. But then $p_1=5$ yields~$U_1(5,7)=\{1,2,3,5\}$.
    \item If~$p_1=6$, then~$p_2=4$ is the best response with~$U_1(6,4)=\{6,8\}$. Again, player~1 improves with~$p_1=7$.
    \item If~$p_1=8$, then~$p_2=6$ is the best response (up to symmetry). Player~1 can improve with~$p_1=5$ which yields~$U_1(5,6)=\{2,3,5\}$.
  \end{compactitem}
\end{proof}

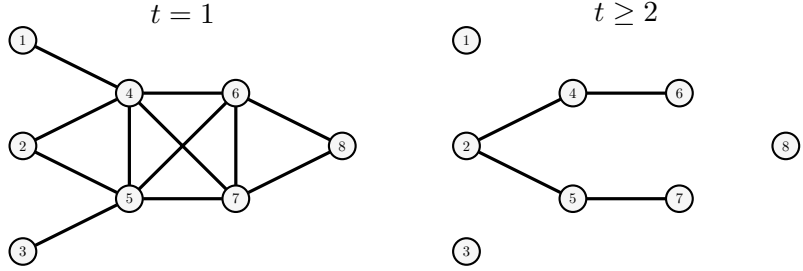
\begin{figure}[t]
  \center 
  \begin{tikzpicture}[scale=.7]
    \node[vertex] (1) at (-2,3) {1};
    \node[vertex] (2) at (-2,1) {2};
    \node[vertex] (3) at (-2,-1) {3};
    \node[vertex] (4) at (0,2) {4};
    \node[vertex] (5) at (0,0) {5};
    \node[vertex] (6) at (2,2) {6};
    \node[vertex] (7) at (2,0) {7};
    \node[vertex] (8) at (4,1) {8};
    \draw[edge] (4) -- (2) -- (5) -- (3);
    \draw[edge] (1) -- (4) -- (5) -- (6) -- (7)--(8) -- (6) -- (4) -- (7) -- (5);
    \node at (1,3.5) {$t=1$};
  \end{tikzpicture}
  \hspace{1cm}
  \begin{tikzpicture}[scale=.7]
    \node[vertex] (1) at (-2,3) {1};
    \node[vertex] (2) at (-2,1) {2};
    \node[vertex] (3) at (-2,-1) {3};
    \node[vertex] (4) at (0,2) {4};
    \node[vertex] (5) at (0,0) {5};
    \node[vertex] (6) at (2,2) {6};
    \node[vertex] (7) at (2,0) {7};
    \node[vertex] (8) at (4,1) {8};
    \draw[edge] (6) -- (4) -- (2) -- (5)-- (7);
    \node at (1,3.5) {$t\ge 2$};
  \end{tikzpicture}
  \caption{A monotonically shrinking temporal split graph without a Nash equilibrium.}
  \label{fig:split}
\end{figure}
\noindent
Note that the classic temporal Voronoi game always has a Nash equilibrium on monotonically shrinking split graphs (as we show in \Cref{thm:Vor-split}). For example, (4,5) is a Nash equilibrium for the temporal split graph in the proof of \Cref{thm:split}.

Contrasting \Cref{thm:split}, we finish this section with a positive result for threshold graphs (and thus also cliques) and complete $k$-partite graphs.

\begin{theorem}\label{thm:k-part-shrink}
  There exists a Nash equilibrium in $\rVor(\mathcal{G},2)$ if $\Gg$ is a monotonically shrinking
  \begin{compactenum}[(1)]
    \item\label{kpart} complete $k$-partite graph with $k\ge 1$ or
    \item\label{thresh} threshold graph.
  \end{compactenum}
\end{theorem}

\begin{proof}
  (\ref{kpart}) Let~$\Gg=(V_1\cup \cdots \cup V_k, (E_t)_{t=1}^{\infty})$ be a monotonically shrinking temporal complete $k$-partite graph with~$k\ge 2$ (the case~$k=1$ is trivial).
  Let~$p_1\in V_1$ and~$p_2\in V_2$. Then, $(p_1,p_2)$ is a Nash equilibrium. Note that~$u_1(p,p_2)=|V_2|$ for all $p\in V\setminus V_2$ and $u_1(p,p_2)\le 1$ if~$p\in V_2$ (and symmetrically for~$u_2(p_1,p)=|V_1|$ if~$p\in V\setminus V_1$ and~$u_2(p_1,p)\le 1$ if~$p\in V_1$) since~$G_1$ is complete~$k$-partite. Hence, no player can improve. 

  (\ref{thresh}) Let~$\Gg=(V, (E_t)_{t=1}^{\infty})$ be a monotonically shrinking temporal threshold graph with~$|V|\ge 2$ (the case $|V|=1$ is trivial). If all vertices are isolated in~$G_1$, then a Nash equilibrium trivially exists. Otherwise, there exists a vertex~$v$ which dominates all non-isolated vertices in~$G_1$. Then, $(v,w)$ with~$w\neq v$ is a Nash equilibrium.
  Clearly, player~1 cannot improve since all non-isolated vertices already reach~$v$ no later than time step~1. Hence, also player~2 cannot improve.
\end{proof}

\section{Temporal Voronoi Games ($\Vor$)}\label{sec:Vor}

We complement the results for the reverse temporal Voronoi game from \Cref{sec:rVor} with
the missing results for the remaining graph classes for the classic temporal Voronoi game.

\subsection{Monotonically Growing Graphs}

We first show that for grids a Nash equilibrium is also not guaranteed.

\begin{theorem}\label{thm:Vor-grid}
  There exists a monotonically growing grid~$\Gg$ such that there is no Nash equilibrium in~$\Vor(\Gg,2)$.
\end{theorem}

\begin{proof}
  Consider the $(3\times4)$-grid~$\Gg$ with vertex set~$[12]$ given in \Cref{fig:Vor-grid}~(left). To see that there is no Nash equilibrium, we consider the best responses (\Cref{fig:Vor-grid}~(right)) which are straightforward to verify.
  \begin{compactitem}
    \item For~$p_1=1$, the best response is~$p_2=6$ with~$u_2(1,6)=10$.
    \item For~$p_1=2$, the best response is~$p_2=6$ with~$u_2(2,6)=5$.
    \item For~$p_1=3$, a best response is~$p_2\in\{6,7\}$ with~$u_2(3,p_2)=8$.
    \item For~$p_1=4$, the best response is~$p_2=3$ with~$u_2(4,3)=9$.
    \item For~$p_1=5$, a best response is~$p_2\in\{2,6,10\}$ with~$u_2(5,p_2)=9$.
    \item For~$p_1=6$, the best response is~$p_2=8$ with~$u_2(6,8)=3$.
    \item For~$p_1=7$, a best response is~$p_2\in\{2,6,10\}$ with~$u_2(7,p_2)=6$.
    \item For~$p_1=8$, the best response is~$p_2=7$ with~$u_2(8,7)=9$.
    \item The case~$p_1\in\{9,10,11,12\}$ is symmetric to~$p_1\in\{1,2,3,4\}$.
  \end{compactitem}
  Note that the above best responses always run into a cycle $6\to 8\to 7\to 6$ or $6\to 8\to 7\to 2 (10) \to 6$. Hence, there exists no Nash equilibrium.
\end{proof}

\begin{figure}[t]
  \center 
  \begin{tikzpicture}[scale=1]
    \node[vertex] (1) at (0,1) {1};
    \node[vertex] (2) at (1,1) {2};
    \node[vertex] (3) at (2,1) {3};
    \node[vertex] (4) at (3,1) {4};
    \node[vertex] (5) at (0,0) {5};
    \node[vertex] (6) at (1,0) {6};
    \node[vertex] (7) at (2,0) {7};
    \node[vertex] (8) at (3,0) {8};
    \node[vertex] (9) at (0,-1) {9};
    \node[vertex] (10) at (1,-1) {10};
    \node[vertex] (11) at (2,-1) {11};
    \node[vertex] (12) at (3,-1) {12};
    \draw[edge] (2) -- (6) -- (10);
    \node at (1.5,1.5) {$t=1$};
  \end{tikzpicture}
  \hspace{1cm}
  \begin{tikzpicture}[scale=1]
    \node[vertex] (1) at (0,1) {1};
    \node[vertex] (2) at (1,1) {2};
    \node[vertex] (3) at (2,1) {3};
    \node[vertex] (4) at (3,1) {4};
    \node[vertex] (5) at (0,0) {5};
    \node[vertex] (6) at (1,0) {6};
    \node[vertex] (7) at (2,0) {7};
    \node[vertex] (8) at (3,0) {8};
    \node[vertex] (9) at (0,-1) {9};
    \node[vertex] (10) at (1,-1) {10};
    \node[vertex] (11) at (2,-1) {11};
    \node[vertex] (12) at (3,-1) {12};
    \draw[edge] (1) -- (2) -- (3) -- (4) -- (8) -- (7) -- (6)-- (5)-- (9)--(10)--(11)--(12);
    \draw[edge] (1) -- (5);
    \draw[edge] (2) -- (6);
    \draw[edge] (3) -- (7);
    \draw[edge] (6) -- (10);
    \draw[edge] (7) -- (11);
    \draw[edge] (8) -- (12);
    \node at (1.5,1.5) {$t\ge 2$};
  \end{tikzpicture}
  \hspace{3cm}
  \begin{tikzpicture}[scale=1]
    \node[vertex] (1) at (0,1) {1};
    \node[vertex] (2) at (1,1) {2};
    \node[vertex] (3) at (2,1) {3};
    \node[vertex] (4) at (3,1) {4};
    \node[vertex] (5) at (0,0) {5};
    \node[vertex] (6) at (1,0) {6};
    \node[vertex] (7) at (2,0) {7};
    \node[vertex] (8) at (3,0) {8};
    \node[vertex] (9) at (0,-1) {9};
    \node[vertex] (10) at (1,-1) {10};
    \node[vertex] (11) at (2,-1) {11};
    \node[vertex] (12) at (3,-1) {12};
    \draw[->,thick] (1)-- (6);
    \draw[->,thick] (2)-- (6);
    \draw[->,thick] (3)-- (6);
    \draw[->,thick] (3)-- (7);
    \draw[->,thick] (4)-- (3);
    \draw[->,thick] (5)-- (6);
    \draw[->,thick] (5)-- (2);
    \draw[->,thick] (5)-- (10);
    \draw[->,thick] (6) to [bend right] (8);
    \draw[->,thick] (7)-- (6);
    \draw[->,thick] (7)-- (2);
    \draw[->,thick] (7)-- (10);
    \draw[->,thick] (8)-- (7);
    \draw[->,thick] (9)-- (6);
    \draw[->,thick] (10) -- (6);
    \draw[->,thick] (11)-- (6);
    \draw[->,thick] (11)-- (7);
    \draw[->,thick] (12)-- (11);
  \end{tikzpicture}
  \caption{(Left) A monotonically growing temporal grid~$\Gg$ without a Nash equilibrium in~$\Vor(\Gg,2)$. (Right) The best response graph of~$\Vor(\Gg,2)$.}
  \label{fig:Vor-grid}
\end{figure}
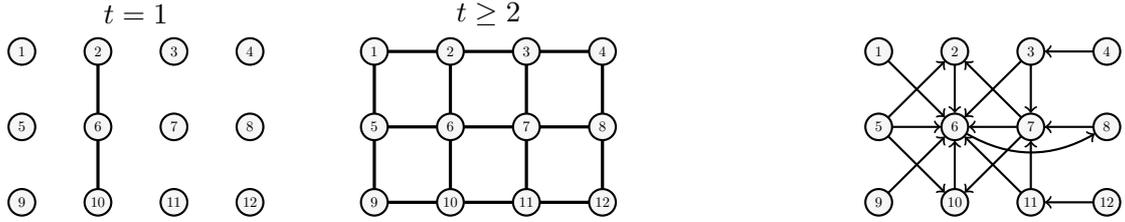

For monotonically growing cliques (and also threshold and split graphs) and complete~$k$-partite graphs, the same arguments as in \Cref{thm:clique-grow,thm:k-part-grow} for the reverse Voronoi game apply.
Hence, from \Cref{thm:Vor-grid}, we also obtain the following.

\begin{corollary}\label{cor:Vor-clique}
  There exists a monotonically growing clique~$\mathcal C$ and a monotonically growing complete~$k$-partite graph~$\mathcal K$ for each~$k\ge 2$ such that~$\Vor(\mathcal C,2)$ and~$\Vor(\mathcal K, 2)$ have no Nash equilibrium.
\end{corollary}

\subsection{Monotonically Shrinking Graphs}

For monotonically shrinking graphs, the following theorem is easily obtained from analogous arguments as for \Cref{thm:k-part-shrink}. Hence, we omit a formal proof.
\begin{theorem}\label{thm:Vor-k-partit}
  For every monotonically shrinking complete~$k$-partite graph~$\mathcal K$ with~$k\ge 1$, there exists a Nash equilibrium in~$\Vor(\mathcal K, 2)$.
\end{theorem}

Finally, for the temporal Voronoi game, a Nash equilibrium is guaranteed even
for monotonically shrinking split graphs (as opposed to the reverse Voronoi game).
Notably, this case is analogous to the Diffusion game on static split graphs~\cite{FHKO22}.

\begin{theorem}\label{thm:Vor-split}
  For every monotonically shrinking split graph~$\mathcal S$, there exists a Nash equilibrium in~$\Vor(\mathcal S, 2)$.
\end{theorem}

\begin{proof}
  Let~$\mathcal S = (V, (E_t)_{t=1}^\infty)$ with~$V=C\cup I$, where~$C$ forms a clique in~$S_1=\mathcal S_\downarrow$ and~$I$ an independent set.
  We assume for each vertex~$v\in I$ that it is not adjacent to all vertices in~$C$, since otherwise we could remove~$v$ from~$I$ and add it to~$C$.
  Hence, we can also assume~$|C|\ge 2$, since otherwise a Nash equilibrium trivially exists.
  We now show that there exists a Nash equilibrium~$(p_1,p_2)$ with~$\{p_1,p_2\}\subseteq C$ and $p_1\neq p_2$.
  To this end, observe that in this case we have~$U_1(p_1,p_2)=\{p_1\}\cup N_I(p_1)\setminus N_I(p_2)$ and~$U_2(p_1,p_2)=\{p_2\}\cup N_I(p_2)\setminus N_I(p_1)$, where~$N_I(v)$ denotes the set of neighbors of~$v$ in~$I$.
  Clearly, no player can improve by choosing a vertex in~$I$ since the payoff then is~$1$.
  Next, we show that the players cannot improve arbitrarily often with vertices in~$C$.
  Assume towards a contradiction that there exists an infinite sequence~$(v_1,w_1),(v_2,w_1),(v_2,w_2),\ldots$ of profiles with~$u_1(v_{i+1},w_i)>u_1(v_i,w_i)$ and~$u_2(v_{i+1},w_{i+1}) > u_2(v_{i+1},w_i)$ for all~$i\ge 1$. Note that this is equivalent to
  \begin{align*}
	|N_I(v_i)| - |N_I(v_i)\cap N_I(w_i)| &< |N_I(v_{i+1})| - |N_I(v_{i+1})\cap N_I(w_i)| \text{ and}\\
	|N_I(w_i)| - |N_I(v_{i+1})\cap N_I(w_i)| &< |N_I(w_{i+1})| - |N_I(v_{i+1})\cap N_I(w_{i+1})|.
  \end{align*}
Since the number of different profiles is finite, there exists a subsequence~$(v_1,w_1),\ldots,(v_i,w_j)$ with~$(v_i,w_j)=(v_1,w_1)$ (wlog).
  But this yields the contradiction
  \begin{align*}
    |N_I(v_1)| + |N_I(w_1)| - |N_I(v_1)\cap N_I(w_1)| &< |N_I(v_2)| + |N_I(w_1)| - |N_I(v_2)\cap N_I(w_1)|\\
                                                      &< |N_I(v_2)| + |N_I(w_2)| - |N_I(v_2) \cap N_I(w_2)|\\
                                                      &< \cdots\\
                                                      &< |N_I(v_i)| + |N_I(w_j)| - |N_I(v_i)\cap N_I(w_j)|\\
    								   &= |N_I(v_1)| + |N_I(w_1)| - |N_I(v_1)\cap N_I(w_1)|.
  \end{align*}
  Hence, there exists a profile where both players cannot improve, that is, a Nash equilibrium.
\end{proof}

\section{Conclusion}\label{sec:conc}

We analyzed Nash equilibria for the classic and the reverse temporal Voronoi game and highlighted some major differences depending on the considered temporal graph.

As regards open questions, note that the classes of temporal graphs we considered already settle the question of guaranteed existence of a Nash equilibrium for most graph classes commonly considered in the literature.
A possible direction for future work would be to further restrict the temporal
behavior of the temporal graph to grow or shrink in a more specific way.
For example, it can be shown that on temporal cycles where at most one edge changes a Nash equilibrium always exists.
Another direction is to study other variants of the temporal Voronoi game.
Here, a natural question is whether Nash equilibria exist for more than two players.
It is also interesting to study the game when the players are allowed to choose more than one vertex initially or if the temporal distance is defined differently (e.g.~with faster arrival instead of earlier).
Finally, it might also be fruitful to investigate the existence of other forms of equilibria, e.g.~when introducing a certain cost for changing.

\bibliographystyle{plainnat}
\newcommand*{\doi}[1]{\href{https://doi.org/#1}{\nolinkurl{doi:#1}}}
\bibliography{ref}

\end{document}